\documentclass{llncs}
\usepackage{amssymb,amsbsy,latexsym}
\usepackage{amsmath}
\usepackage{graphics, subfigure, float}
\usepackage[applemac]{inputenc}
\usepackage[american]{babel}
\usepackage[T1]{fontenc} 
\usepackage{ae,aecompl} 
\usepackage[displaymath,textmath,sections]{preview} 
\usepackage{pst-all}
\usepackage{pstricks-add}
\providecommand{\Q}{\mathbb{Q}}

\newcommand{\eps}{\varepsilon}

\floatstyle{ruled}
\newfloat{Figure}{t}{fig}
\newfloat{algorithm}{t}{}
\floatname{algorithm}{Algorithm}

\PreviewEnvironment{enumerate} % Attention: Do not use comma-seperated lists here 
\PreviewEnvironment{itemize}
\PreviewEnvironment{theorem} 
\PreviewEnvironment{corollary} 
\PreviewEnvironment{lemma} 
\PreviewEnvironment{myproblem}
\PreviewEnvironment{abstract} 
\PreviewEnvironment{defn} 
\PreviewEnvironment{center} 
\PreviewEnvironment{pspicture} 
\psset{arrowsize=6pt, labelsep=2pt, linewidth=1.5pt}
\makeatother

\title{On Min-Power Steiner Tree\thanks{This research is supported by the ERC Starting Grant NEWNET 279352.}}
\author{Fabrizio Grandoni}
\institute{IDSIA, University of Lugano, Switzerland, {\tt fabrizio@idsia.ch}}

\begin{document}

\maketitle

\vspace{-20pt}

\begin{abstract}
\noindent In the classical (min-cost) Steiner tree problem, we are given an edge-weighted undirected graph and a set of terminal nodes. The goal is to compute a min-cost tree $S$ which spans all  terminals. In this paper we consider the min-power version of the problem (a.k.a. symmetric multicast), which is better suited for wireless applications. Here, the goal is to minimize the total power consumption of nodes, where the power of a node $v$ is the maximum cost of any edge of $S$ incident to $v$. Intuitively, nodes are antennas (part of which are terminals that we need to connect) and edge costs define the power to connect their endpoints via bidirectional links (so as to support protocols with ack messages). Observe that we do not require that edge costs reflect Euclidean distances between nodes: this way we can model obstacles, limited transmitting power, non-omnidirectional antennas etc. Differently from its min-cost counterpart, min-power Steiner tree is NP-hard even in the spanning tree case (a.k.a. symmetric connectivity), i.e. when all nodes are terminals. Since the power of any tree is within once and twice its cost, computing a $\rho_{st}\leq \ln(4)+\eps$ [Byrka et al.'10] approximate min-cost Steiner tree provides a $2\rho_{st}<2.78$ approximation for the problem. For min-power spanning tree the same approach provides a $2$ approximation, which was improved to $5/3+\eps$ with a non-trivial approach in [Althaus et al.'06].  

\quad  In this paper we present an improved approximation algorithm for min-power Steiner tree. Our result is based on two main ingredients. We present the first decomposition theorem for min-power Steiner tree, in the spirit of analogous structural results for min-cost Steiner tree and min-power spanning tree. Based on this theorem, we define a proper LP relaxation, that we exploit within the iterative randomized rounding framework in [Byrka et al.'10]. A careful analysis of the decrease of the power of nodes at each iteration provides a $3\ln 4-\frac{9}{4}+\eps<1.91$ approximation factor. The same approach gives an improved $1.5+\eps$ approximation for min-power spanning tree as well. This matches the approximation factor in [Nutov and Yaroshevitch'09] for the special case of min-power spanning tree with edge weights in $\{0,1\}$.

%\vspace{0.2cm}
%\noindent{\bf Keywords:} Approximation algorithms; Power minimization; Steiner tree; Symmetric multicast and connectivity.
\end{abstract}

\vspace{-30pt}

\section{Introduction}
\label{sec:intro}

%\cite{ACMPTZ06,BGRS10,NY09}

Consider the following basic problem in wireless network design. We are given a set of antennas, and we have to assign the transmitting power of each antenna. Two antennas can exchange messages (directly) if they are within the transmission range of each other (this models protocols with ack messages). The goal is to find a minimum total power assignment so that a given subset of antennas can communicate with each other (using a multi-hop protocol). 

We can formulate the above scenario as a \emph{min-power Steiner tree} problem (a.k.a. \emph{symmetric multicast}). Here we are given an undirected graph $G=(V,E)$, with edge costs $c:E\to \Q_{\geq 0}$, and a subset $R$ of terminal nodes. The goal is to compute a Steiner tree $S$ spanning $R$, of minimum \emph{power} $p(S):=\sum_{v\in V(S)}p_S(v)$, with $p_S(v):=\max_{uv\in E(S)}\{c(uv)\}$. In words, the power of a node $v$ with respect to tree $S$ is the largest cost of any edge of $S$ incident to $v$, and the power of $S$ is the sum of the powers of its nodes\footnote{When $S$ is clear from the context, we will simply write $p(v)$.}.
The {\em min-power spanning tree} problem (a.k.a. \emph{symmetric connectivity}) is the special case of min-power Steiner tree where $R=V$, i.e. all nodes are terminals. Let us remark that, differently from part of the literature on related topics, we do not require that edge costs reflect Euclidean distances between nodes. This way, we are able to model obstacles, limited transmitting power, antennas which are not omnidirectional, etc.
%\footnote{This problem is also known as the min-power connectivity problem. Throughout this paper we name min-power problems following the name of the associated min-cost version.}.
%Let us remark that, differently from the vast majority of the literature on the topic, we do not require that edge weights reflect Euclidean distances between antennas. Indeed, the latter assumption is fairly unrealistic due to obstacles, weather conditions, etc. So, it can happen that more power is required to connect pairs of antennas which are geographically closer.

The following simple approximation-preserving reduction shows that min-power Steiner tree is at least as hard to approximate as its min-cost counterpart: given a min-cost Steiner tree instance, replace each edge $e$ with a path of $3$ edges, where the boundary edges have cost zero and the middle one has cost $c(e)/2$. Hence the best we can hope for in polynomial time is a $c$ approximation for some constant $c>1$. It is known \cite{ACMPTZ06,KKKP00} that, for any tree $S$ of cost $c(S):=\sum_{e\in E(S)}c(e)$,
\begin{equation}\label{eqn:2factor}
c(S)\leq p(S)\leq 2c(S).
\end{equation}
As a consequence, a $\rho_{st}$ approximation for min-cost Steiner tree implies a $2\rho_{st}$ approximation for min-power Steiner tree. In particular, the recent $\ln(4)+\eps<1.39$ approximation\footnote{Throughout this paper $\eps$ denotes a small positive constant.} in \cite{BGRS10} for the first problem, implies a $2.78$ approximation for the second one: no better approximation algorithm is known to the best of our knowledge. 

Differently from its min-cost version, min-power spanning tree is NP-hard (even in quite restricted subcases) \cite{ACMPTZ06,KKKP00}.
By the above argument, a min-cost spanning tree is a $2$ approximation. However, in this case non-trivial algorithms are known. A $1+\ln 2+\eps<1.69$ approximation is given in \cite{CMZ02}. This was improved to $\frac{5}{3}+\eps$ in \cite{ACMPTZ06}. If edge costs are either $0$ or $1$, the approximation factor can be further improved to $\frac{3}{2}+\eps$ \cite{NY09}. Indeed, the same factor can be achieved if edge costs are either $a$ or $b$, with $0\leq a<b$: this models nodes with two power states, low and high. All these results exploit the notion of $k$-decomposition. The first result is obtained with a greedy algorithm, while the latter two use (as a black box) the FPTAS in \cite{PS00} for the \emph{min-cost connected spanning hypergraph} problem in $3$-hypergraphs. We will also use $k$-decompositions, but our algorithms are rather different (in particular, they are LP-based).

%However, in this case a non-trivial, improved 
%$\frac{5}{3}+\eps$ approximation is known \cite{ACMPTZ06}. The basic idea is that one can decompose the optimal Steiner tree into subtrees of at most $3$ nodes each, while increasing the cost at most by a factor $5/3$. Finding the cheapest $3$-decomposition is a special case of the \emph{min-cost connected spanning hypergraph} problem in $3$-hypergraphs: for the latter problem an FPTAS is known \cite{PS00}. If edge costs are either $0$ or $1$, the above factor $5/3$ can be reduced to $3/2$, implying a $\frac{3}{2}+\eps$ approximation for this special case \cite{NY09}.
%If edge costs are either $0$ or $1$, the approximation for min-power spanning tree can be further improved to $\frac{3}{2}+\eps$ \cite{NY09}.

\paragraph{\bf Our Results.} In this paper we present an improved approximation algorithm for min-power Steiner tree. 
\begin{theorem}\label{thr:mainSteiner}
There is an expected $3\ln 4-\frac{9}{4}+\eps<1.909$ approximation algorithm for min-power Steiner tree.
\end{theorem}
Our result is based on two main ingredients. Informally, a $k$-decomposition of a Steiner tree $S$ is a collection of (possibly overlapping) subtrees of $S$, each one containing at most $k$ terminals, which together span $S$. The power/cost of a decomposition is the sum of the powers/costs of its components\footnote{Due to edge duplication, the cost of the decomposition can be larger than $c(S)$. Its power can be larger than $p(S)$ even for edge disjoints components.}. It is a well-known fact (see \cite{BD97} and references therein) that, for a constant $k$ large enough, there exists a $k$-decomposition of cost at most $(1+\eps)$ times the cost of $S$. A similar result holds for min-power spanning tree \cite{ACMPTZ06}. The first ingredient in our approximation algorithm is a similar decomposition theorem for min-power Steiner tree, which might be of independent interest. This extends the qualitative results in \cite{ACMPTZ06,BD97} since min-power Steiner tree generalizes the other two problems. However, the dependence between $\eps$ and $k$ is worse in our construction.
%First of all, we prove a decomposition theorem for min-power Steiner tree. Informally, an almost optimal solution can be expressed as the combination of \emph{components}, each one containing a constant number of terminals (formal definitions are given in the technical part). This is crucial to obtain our improved approximation algorithm, but it might be also of independent interest.
\begin{theorem}\label{thr:decomposition} {\bf (Decomposition)}
For any $h\geq 3$ and any Steiner tree $S$, there exists a $h^h$-decomposition of $S$ of power at most $(1+\frac{14}{h})p(S)$.
\end{theorem} 
Based on this theorem, we are able to compute a $1+\eps$ approximate solution for a proper component-based LP-relaxation for the problem. We exploit this relaxation within the iterative randomized rounding algorithmic framework in \cite{BGRS10}: we sample one component with probability proportional to its fractional value, set the corresponding edge costs to zero and iterate until there exists a Steiner tree of power zero. The solution is given by the sampled components plus a subset of edges of cost zero in the original graph. A careful analysis of the decrease of node powers at each iteration provides a $3\ln 4-\frac{9}{4}+\eps<1.91$ approximation. We remark that, to the best of our knowledge, this is the only other known application of iterative randomized rounding to a natural problem. 
%Observe that for several natural problem in this area, improving on a $2$ approximation is non-trivial.

The same basic approach also provides an improved approximation for min-power spanning tree. 
\begin{theorem}\label{thr:mainSpanning}
There is an expected $\frac{3}{2}+\eps$ approximation algorithm for min-power spanning tree.
\end{theorem}
This improves on \cite{PS00}, and matches the approximation factor achieved in \cite{NY09} (with a drastically different approach!) for the special case of $0$-$1$ edge costs.

%\paragraph{Our Results.} 

%\subsection{Our Results}
%\label{sec:results}

\paragraph{\bf Preliminaries and Related Work.} Min-power problems are well studied in the literature on wireless applications. Very often here one makes the assumption that nodes are points in $\mathbb{R}^2$ or $\mathbb{R}^3$, and that edge costs reflect the Euclidean distance $d$ between pairs of nodes, possibly according to some power law (i.e., the cost of the edge is $d^c$ for some constant $c$ typically between $2$ and $4$). This assumption is often not realistic for several reasons. First of all, due to obstacles, connecting in a direct way geographically closer nodes might be more expensive (or impossible). Second, the power of a given antenna might be upper bounded (or even lower bounded) for technological reasons. Third, antennas might not be omnidirectional. All these scenarios are captured by the undirected graph model that we consider in this paper. A relevant special case of the undirected graph model is obtained by assuming that there are only two edge costs $a$ and $b$, $0\leq a < b$. This captures the practically relevant case that each node has only two power states, low and high. 
%Let us remark that there is an even more general setting, modeled via directed graphs, where the connection cost is \emph{asymmetric}: this captures phenomena like reflections as well.
A typical goal is to satisfy a given connectivity requirement at minimum total power, as we assume in this paper. However, it makes sense also to consider the min-max version of the problem, where one wants to minimize the maximum power. 

Several results are known in the \emph{asymmetric} case, where a unidirectional link is established from $u$ to $v$ iff $v$ is within the transmission range of $u$ (and possibly the vice versa does not hold). For example in the \emph{asymmetric unicast} problem one wants to compute a min-power directed path from node $s$ to node $t$. This problem can be solved in polynomial time, say, via dynamic programming. In the \emph{asymmetric connectivity} problem one wants to compute a min-power spanning arborescence rooted at a given root node $r$. This problem is NP-hard even in the $2$-dimensional Euclidean case 
\cite{CCPRV01}, and a minimum spanning tree provides a $12$ approximation for the Euclidean case (while the general case is log-hard to approximate) \cite{WCLF01}. The \emph{asymmetric multicast} problem is the generalization of asymmetric connectivity where one wants a min-power arborescence rooted at $r$ which contains a given set $R$ of terminals. As observed in \cite{ACMPTZ06}, the same approach as in \cite{WCLF01} provides a $12\rho_{st}$ approximation for the Euclidean case, where $\rho_{st}$ is the best-known approximation for Steiner tree in graphs.
In the \emph{complete range assignment} problem one wants to establish a strongly connected spanning subgraph. The authors of \cite{KKKP00} present a $2$-approximation which works for the undirected graph model, 
%(just compute a minimum spanning tree and bidirect it), 
and show that the problem is NP-hard in the $3$-dimensional Euclidean case. The NP-hardness proof was extended to $2$ dimensions in \cite{CPS00}. Recently, the approximation factor was improved to $2-\delta$ for a small constant $\delta>0$ with a highly non-trivial approach \cite{C10}. The same paper presents a $1.61$ approximation for edge costs in $\{a,b\}$,  improving on the $9/5$ factor in \cite{CK07}.

In this paper we consider the symmetric case, where links must be bidirectional (i.e. $u$ and $v$ are not adjacent if one of the two is not able to reach the other). This is used to model protocols with ack messages. The symmetric unicast problem can be solved by applying Dijkstra's algorithm to a proper auxiliary graph \cite{ACMPTZ06}. The symmetric connectivity and multicast problems are equivalent to min-power spanning tree and min-power Steiner tree, respectively. We already discussed the known results on these problems. One can also consider higher connectivity requirements. For example, $O(k)$ \cite{HIM07} and $O(\log^4 n)$ \cite{HKMN07} approximation algorithms are known for the problem 
of computing a min-power $k$-vertex connected subgraph (with bidirectional links).

%As noted in \cite{ACMPTZ06}, by \eqref{eqn:2factor} a min-cost spanning tree is a $2$-approximation for the \emph{symmetric broadcast} problem (i.e. the min-power spanning tree problem in our terminology). This is because, for any tree $T$ of power $p(T)$ and cost $c(T)$,
%%\begin{equation}
%$c(T)\leq p(T)\leq 2c(T)$.
%%\label{eqn:2factor}
%%\end{equation}
%An improved $5/3+\eps$ approximation is given in \cite{ACMPTZ06}. 
%This approximation is based on the notion of $k$-decomposition, i.e. a collection of trees on at most $k$ nodes each, that together span all the nodes. The authors show that there exists a $3$-decomposition which $5/3$ approximates the optimum. The $5/3+\eps$ approximation then follows by observing that finding the cheapest $3$-decomposition is equivalent to solving the \emph{min-cost connected spanning hypergraph} problem in $3$-hypergraphs: for the latter problem an FPTAS is known \cite{PS00}. The same approach gives a $3/2+\eps$ approximation in the case of $0$-$1$ costs, by showing that in this case $3$-decompositions are $3/2$ approximate \cite{NY09}. As we will see, our $1.57$ approximation for min-power spanning tree also exploits $k$-decompositions, but it is rather different from an algorithmic point of view. The \emph{symmetric multicast} problem (min-power Steiner tree in our terminology) was not addressed before for general graphs to the best of our knowledge. However, by Equation \eqref{eqn:2factor}, a $\rho_{st}\leq \ln 4 +\eps$ approximate min-cost Steiner tree provides a $2\rho_{st}<2.78$ approximation. 

Proofs which are omitted due to lack of space are given in the appendix. The min-power Steiner tree is denoted by $S^*$.

\section{A Decomposition Theorem for Min-Power Steiner Tree}
\label{sec:decomposition}

A $k$-\emph{component} is a tree which contains at most $k$ terminals. If internal nodes are non-terminals, the component is \emph{full}. A $k$-\emph{decomposition} of a Steiner tree $S$ over terminals $R$ is a collection of $k$-components on the edges of $S$ which span $S$ and such that the following auxiliary \emph{component graph} is a tree: replace each component $C$ with a star, where the leaves are the terminals of $C$ and the central node is a distinct, dummy non-terminal $v_C$. Observe that, even if the component graph is a tree, the actual components might share edges. When the value of $k$ is irrelevant, we simply use the terms component and decomposition.
We will consider $k$-decompositions with $k=O(1)$. This is useful since a min-power component $C$ on a constant number of terminals can be computed in polynomial time\footnote{One can guess (by exhaustive enumeration) the non-terminal nodes of degree at least $3$ in $C$, and the structure of the tree where non-terminals of degree $2$ are contracted. Each edge $vu$ of the contracted tree corresponds to a path $P$ whose internal nodes are non-terminals of degree $2$: after guessing the boundary edges $uu'$ and $v'v$ of $P$ (which might affect the power of $u$ and $v$, respectively), the rest of $P$ is w.l.o.g. a min-power path between $u'$ and $v'$ (which can be computed in polynomial time \cite{ACMPTZ06}).}.
%how such nodes are connected among them and with terminals by paths whose internal nodes have degree $2$ in $C$. One can also guess the boundary edges $e'$ and $e''$ of each such path: given that, the entire path can be derived by solving a symmetric unicast problem in polynomial time \cite{ACMPTZ06}.}. 
The assumption on the component graph is more technical, and it will be clearer later. Intuitively, when we compute a min-power component on a subset of terminals, we do not have full control on the internal structure of the component. For this reason, the connectivity requirements must be satisfied independently from that structure.
%For this reason, allowing connectivity among components via Steiner nodes is risky.

%The power of a decomposition is the sum of the power of its components. It is a well-know fact (see \cite{BD97} and references therein) that, for a constant $k$ large enough, there exists a $k$-decomposition of cost (i.e. sum of the costs of its components) at most $1+\eps$ times the minimum cost of a Steiner tree. A similar result holds for min-power spanning tree \cite{ACMPTZ06}. In this section we present the first structural result of this type for min-power Steiner tree. This extends the qualitative results in \cite{ACMPTZ06,BD97} since min-power Steiner tree generalizes the other two problems. However, the dependence between $\eps$ and $k$ is worse in our construction.

Assume w.l.o.g. that $S$ consists of one full component. This can be enforced by appending to each terminal $v$ a dummy node $v'$ with a dummy edge of cost $0$, and replacing $v$ with $v'$ in the set of terminals. Any decomposition into (full) $k$-components of the resulting tree can be turned into a $k$-decomposition of the same power for $S$ by contracting dummy edges, and vice versa.

%W.l.o.g., we can assume that any Steiner tree consists of exactly one component. In fact, an equivalent instance of the problem is obtained by appending to each terminal $v$ a dummy node $v'$ with a dummy edge of cost $0$, and replacing $v$ with $v'$ in the set of terminals.

Next lemma shows that one can assume that the maximum degree of the components in a decomposition can be upper bounded by a constant while losing a small factor in the approximation (see also Figure \ref{fig:lemDecomposition}). 
\begin{lemma}\label{lem:boundedDegree}
For any $\Delta\geq 3$, there exists a  decomposition of $S$ of power at most 
$(1+\frac{2}{\lceil\Delta/2\rceil-1})p(S)$ 
%$\frac{\lceil\Delta/2\rceil+1}{\lceil\Delta/2\rceil-1}p(S)$ 
whose components have degree at most $\Delta$.
\end{lemma}
\begin{proof}
%We construct the decomposition iteratively. 
The rough idea is to split $S$ at some node $v$ not satisfying the degree constraint, so that the duplicated copies of $v$ in each obtained component have degree (less than) $\Delta$. Then we add a few paths between components, so that the component graph remains a tree. All the components but one will satisfy the degree constraint: we iterate the process on the latter component. 
%We implement the process according to a proper order on the nodes, so that we can bound the final power via a charging argument.  

In more detail, choose any leaf node $r$ as a \emph{root}. The decomposition initially consists of $S$ only. We maintain the invariant that all the components but possibly the component $C_r$ containing $r$ have degree at most $\Delta$. Assume that $C_r$ has degree larger than $\Delta$ (otherwise, we are done). Consider any \emph{split node} $v$ of degree $d(v)=d+1\geq \Delta+1$ such that all its descendants have degree at most $\Delta$. Let $u_1,\ldots,u_{d}$ be the children of $v$, in increasing order of $c(vu_i)$. Define $\Delta':=\lceil \Delta/2\rceil\in [2,\Delta-1]$. We partition the $u_i$'s by iteratively removing the first $\Delta'$ children, until there are at most $\Delta-2$ children left: let $V_1,\ldots,V_h$ be the resulting subsets of children. In particular, for $i<h$, $V_i=\{u_{(i-1)\Delta'+1},\ldots,u_{i\Delta'} \}$. For $i=1,\ldots,h-1$, we let $C_i$ be a new component induced by $v$, $V_i$, and the descendants of $V_i$. The new root component $C_h$ is obtained by removing from $C_r$ the nodes $\cup_{i<h}V(C_i)-\{v\}$ and the corresponding edges. In order to 
maintain the connectivity of the component graph (which might be lost at this point), we expand $C_{i+1}$, $i\geq 1$, as follows: let $P_j$ be any path from $v$ to some leaf which starts with edge $vu_j$. We append to $C_{i+1}$ the path $P_{m(i)}$ which minimizes $p(P_j)-c(vu_j)$ over $j\in V_i$.
%for each $u_j\in V_{i}$, let $P_j$ be any path (say, the min-power one) between $v$ and some leaf descendant of $u_j$ (in particular, $P_j$ will start from edge $vu_j$). We select the path $P_{m(i)}$ which minimizes $p(P_j)-c(vu_j)$ among the $P_j$'s, and append $P_{m(i)}$ to $C_{i+1}$. 
After this step, the component graph is a tree. The invariant is maintained: in fact, the degree of any node other than $v$ can only decrease. In each $C_i$, $i<h$, $v$ has degree either $\Delta'$ or $\Delta'+1$, which is within $2$ and $\Delta$. Since $\Delta-2-\Delta'< |V_h|\leq \Delta-2$, the cardinality $|V_h|+2$ of $v$ in $C_h$ is also in $[2,\Delta]$.
%Node $v$ has in each component degree at most $\max\{\Delta',\Delta'+1,\Delta-2+2\}\leq \Delta$ and at least $\min\{\Delta',(\Delta-\Delta'-1)+2\}\geq 2$. 
By the choice of $v$, all the components but $C_r$ have maximum degree $\Delta$. Observe that $C_r$ loses at least $\Delta'-1\geq 1$ nodes, hence the process halts.
%$d(w)$ of all nodes $w$ remains the same, but possibly for duplicated nodes. The degree of the split node $v$ satisfies $\Delta'\leq d(v)\leq \Delta'+1\leq \Delta$ in $C_i$, $i<h$,  and $\Delta'\leq \Delta-1-\Delta'+2\leq d(v)\leq \Delta-2+2=\Delta$ in $C_h$. Internal nodes in $P_{m(i)}$, $i<h$, have degree $2$ in $C_{i+1}$, and maintain the original degree in $C_{i}$. Observe that $C_r$ loses at least $\Delta'-1\geq 1$ nodes, hence the process halts. %in a finite (indeed polynomial) number of steps.

In order to bound the power of the final decomposition, we use the following charging argument, consisting of two charging rules.
When we split $C_r$ at a given node $v$, the power of $v$ remains $p(v)$ in $C_h$ and becomes  $c(vu_{i\Delta'})$ in the other $C_i$'s. We evenly charge the extra power $c(vu_{i\Delta'})$ to nodes $V_{i+1}$: observe that each $u_j\in V_{i+1}$ is charged by $\frac{c(vu_{i\Delta'})}{|V_{i+1}|}\leq \frac{c(vu_{i\Delta'})}{\Delta'-1}\leq \frac{c(vu_{j})}{\Delta'-1}\leq \frac{p(u_{j})}{\Delta'-1}$.
%We charge the extra power $c(vu_{i\Delta'})$ to nodes $V_{i+1}$, so that no node $u_j$ is charged by more than $\frac{1}{\Delta'-1}p(u_j)$. This is doable since $\frac{c(vu_{i\Delta'})}{|V_{i+1}|}\leq \frac{c(vu_{i\Delta'})}{\Delta'-1}\leq \frac{c(vu_{j})}{\Delta'-1}\leq \frac{p(u_{j})}{\Delta'-1}$.

Furthermore, we have an extra increase of the power by $p(P_{m(i)})-c(vu_{m(i)})$ for every $i<h$: this is charged to the nodes of the paths $P_j-\{v\}$ with $u_j\in V_{i}-\{u_{m(i)}\}$, in such a way that no node $w$ is charged by more than $\frac{1}{\Delta'-1}p(w)$. This is possible since there are $\Delta'-1$ such paths, and the nodes of each such path have total power at least $p(P_{m(i)})-c(vu_{m(i)})$ by construction.
%{\small
%$$
%p(P_{m(i)}-\{v\})\leq \sum_{\substack{(i-1)\Delta'\leq j\leq i\Delta'\\j\neq m(i)}}\frac{p(P_{j}-\{v\})}{\Delta'-1}\leq \sum_{\substack{(i-1)\Delta'\leq j\leq i\Delta'\\j\neq m(i)}}\,\sum_{w\in V(P_j)-\{v\}}\frac{p(w)}{\Delta'-1}.
%$$
%}
%\hspace{-4pt}

Each node $w$ can be charged with the second charging rule at most once, since when this happens $w$ is removed from $C_r$ and not considered any longer. When $w$ is charged with the first charging rule, it must be a child of some split node $v$. Since no node is a split node more than once, also in this case we charge $w$ at most once. Altogether, each node $v$ is charged by at most $\frac{2}{\Delta'-1}p(v)$.\qed 
%Altogether, the cost of the decomposition is at most
%$\sum_{v}(1+\frac{2}{\Delta'-1})p(v)= \frac{\lceil\Delta/2\rceil+1}{\lceil\Delta/2\rceil-1}p(S^*).$
\end{proof}

%\begin{theorem}\label{thr:decomposition} {\bf (Decomposition)}
%For any $h\geq 3$, there exists an $h^h$-decomposition of $S$ of power at most $(1+\frac{14}{h})p(S)$.
%\end{theorem} 
\begin{proof} {\em (Theorem \ref{thr:decomposition})}
Apply Lemma \ref{lem:boundedDegree} with $\Delta=h$ to $S$, hence obtaining a decomposition  of power at most $\frac{\lceil h/2\rceil+1}{\lceil h/2\rceil-1}p(S)$ whose components have degree at most $h$. We describe an $h^h$ decomposition of each such component $C$ with more than $h^h$ terminals (see also Figure \ref{fig:thrDecomposition}). Root $C$ at any non-terminal $r$, and shortcut internal nodes (other than $r$) of degree $2$. For any internal node $v$ of $C$, let $P(v)$ be the path from $v$ to its rightmost child $r_v$, and then from $r_v$ to some leaf terminal $\ell(v)$ using the leftmost possible path. Observe that paths $P(v)$ are edge disjoint. Pick a value $q\in \{0,1,\ldots,h-1\}$ uniformly at random, and mark the nodes at level $\ell = q \pmod h$. Consider the partition of $C$ into edge-disjoint subtrees $T$ which is induced by the marked levels.
%We use marked nodes to partition $C$ into subtrees (i.e., we consider maximal subtrees between consecutive marked levels). 
Finally, for each such subtree $T$, we append to each leaf $v$ of $T$ the path $P(v)$: this defines a component $C_T$.

Trees $T$ have at most $h^h$ leaves: hence components $C_T$ contain at most $h^h$ terminals each. Observe that the component graph remains a tree. In order to bound the power of components $C_T$, note that each node $u$ in the original tree has in each component a power not larger than the original power $p(u)$: hence it is sufficient to bound the expected number $\mu_{u}$ of components a node $u$ belongs to. 
%Observe that $\mu_r=1$, hence we can assume $u\neq r$. 
Suppose $u$ is contracted or a leaf node. Then $u$ is contained in precisely the same components as some edge $e$. This edge belongs deterministically to one subtree $T$ (hence to $C_T$), and possibly to another component $C_{T'}$ if the node $v$ with $e\in P(v)$ is marked: the latter event happens with probability $1/h$. Hence in this case $\mu_u\leq 1+1/h$. For each other node $u$, observe that  $u$ belongs to one subtree $T$ if it is not marked, and to at most two such subtrees otherwise. Furthermore, it might belong to one extra component $C_{T'}$ if the node $v$ with $ul_u\in P(v)$ is marked, where $l_u$ is the leftmost child of $u$. Hence, $\mu_u\leq 1+2/h$ in this case. Altogether, the decomposition of $C$ has power at most $(1+2/h)p(C)$ in expectation. 

From the above discussion, there exists (deterministically) an $h^h$ decomposition of power at most $\frac{\lceil h/2\rceil+1}{\lceil h/2\rceil-1}(1+\frac{2}{h})p(S)\leq (1+\frac{14}{h})p(S)$.\qed 
\end{proof}
We remark that for both min-cost Steiner tree and min-power spanning tree (which are special cases of min-power Steiner tree), improved $(1+\frac{O(1)}{h})$-approximate $c^{h}$ decompositions, $c=O(1)$, are known \cite{ACMPTZ06,BD97}. Finding a similar result for min-power Steiner tree, if possible, is an interesting open problem in our opinion (even if it would not directly imply any improvement of our approximation factor).
% Note however that such a result would not imply an improvement of the approximation factor with our approach.

%One of the key ingredients in our construction is a decomposition theorem for the first problem, which might be of independent interest. For our purposes, any $1+\frac{O(1)}{h}$ approximate $f(h)$-decomposition is sufficient. However, it would be interesting to find a (substantially) stronger decomposition theorem, or show that this is not possible. For example, is there any such decomposition with $f(h)=c^{h}$ for some constant $c$? 

\section{An Iterative Randomized Rounding Algorithm}
\label{sec:improved}

In this section we present an improved approximation algorithm for min-power Steiner tree. Our approach is highly indebted to \cite{BGRS10}. 
%\paragraph{\bf A Component-Based Relaxation.} 
We consider the following LP relaxation for the problem:
{\small
\begin{align*}
\min & \sum_{(Q,s): s\in Q\subseteq R}p_{Q}\cdot x_{Q,s} & (LP_{pow}) \\
s.t.   & \sum_{\substack{(Q,s): s\in Q\subseteq R,\\ s\notin W,Q\cap W\neq \emptyset}}x_{Q,s}\geq 1, & \forall \emptyset\neq W\subseteq R-\{r\};\\
      & x_{Q,s}\geq 0, & \forall s\in Q\subseteq R.
\end{align*}
}
\hspace{-4pt}Here $r$ is an arbitrary \emph{root} terminal. There is a variable $x_{Q,s}$ for each subset of terminals $Q$ and for each $s\in Q$: the associated coefficient $p_Q$ is the power of a min-power component $C_{Q}$ on terminals $Q$. In particular, $S^*=C_{R}$ induces a feasible integral solution (where the only non-zero variable is $x_{R,r}=1$). Let 
$C_{Q,s}$ be the \emph{directed component} which is obtained by directing the edges of $C_{Q}$ towards $s$. For a fractional solution $x$, let us define a directed capacity reservation by
considering each $(Q,s)$, and increasing by  $x_{Q,s}$ the capacity of the edges in $C_{Q,s}$ 
% directing $C_{Q}$ towards $s$ and assigning capacity $x_{Q,s}$ to its directed edges in a cumulative way. 
Then the cut constraints ensure that each terminal is able to send one unit of (splittable) flow to the root without exceeding the mentioned capacity reservation. We remark that the authors of \cite{BGRS10} consider essentially the same LP, the main difference being that $p_{Q}$ is replaced by the cost $c_{Q}$ of a min-cost component on terminals $Q$\footnote{Another technical difference w.r.t. \cite{BGRS10} is that they consider only full components: this has no substantial impact on their analysis, and allows us to address the Steiner and spanning tree cases in a unified way.}. In particular, the set of constraints in their LP is the same as in $LP_{pow}$. This allows us to reuse part of their results and techniques, which rely only on the properties of the set of constraints\footnote{Incidentally, this observation might be used to address also other variants of the Steiner tree problem, with different objective functions.}. 

Given Theorem \ref{thr:decomposition}, the proof of the following lemma follows along the same line as in \cite{BGRS10}.
%Observe that $LP_{pow}$ has an exponential number of variables and constraints: still, we can compute an almost optimal solution to it thanks to Theorem \ref{thr:decomposition}.
\begin{lemma}\label{lem:polyLP}
For any constant $\eps>0$, a $1+\eps$ approximate solution to $LP_{pow}$ can be computed in polynomial time.
\end{lemma} 

%\paragraph{\bf Iterative Randomized Rounding.} 

We exploit $LP_{pow}$ within the iterative randomized rounding framework in \cite{BGRS10}. Our algorithm (see also Algorithm \ref{alg:mpst}) consists of a set of iterations. At each iteration $t$ we compute a $(1+\eps)$-approximate solution to $LP_{pow}$, and then sample one component $C^t=C_{Q}$ with probability proportional to $\sum_{s\in Q}x^t_{Q,s}$. We set to zero the cost of the edges of $C^t$ in the graph, updating $LP_{pow}$ consequently\footnote{In the original algorithm in \cite{BGRS10}, the authors contract components rather than setting to zero the cost of their edges. Our variant has no substantial impact on their analysis, but it is crucial for us since contracting one edge (even if it has cost zero) can decrease the power of the solution.}. The algorithm halts when there exists a Steiner tree of cost (and power) zero: this halting condition can be checked in polynomial time. 
%The output solution is any Steiner tree in the union of the sampled components plus the edges of cost zero in the original graph.
\begin{algorithm}
\vspace{5pt}
\begin{enumerate}
\item[(1)] For $t=1,2,\ldots$ 
\item[{}] \begin{enumerate}
  \item[(1a)] Compute a $1+\eps$ approximate solution $x^t$ to $LP_{pow}$ (w.r.t. the current instance).
  \item[(1b)] Sample one component $C^t$, where $C^t=C_{Q}$ with probability $\sum_{s\in Q}x^t_{Q,s}/\sum_{(Q',s')}x^t_{Q',s'}$. Set to zero the cost of the edges in $C^t$ and update $LP_{pow}$.
  \item[(1c)] If there exists a Steiner tree of power zero, return it and halt. 
%  any Steiner tree in the union of the sampled components plus the edges of original cost zero.
\end{enumerate}
\end{enumerate}
%\item[(2)] Compute a terminal spanning tree $T^\mu$ in the remaining instance.
%\item[(3)] Output $T^\mu \cup \bigcup_{t=1}^\mu C^t$.
%\end{enumerate*}
\caption{An iterative randomized rounding approximation algorithm for min-power Steiner tree.}\label{alg:mpst} 
\end{algorithm}

\begin{lemma}\label{lem:time}
Algorithm \ref{alg:mpst} halts in a polynomial number of rounds in expectation.
\end{lemma}

\section{An Improved Approximation.}

In this section we bound the approximation factor of Algorithm \ref{alg:mpst}, 
both in the general and in the spanning tree case.
%starting with the simpler spanning tree case, and then proceeding with the more technical analysis for the general case. 
%We next show that Algorithm \ref{alg:mpst} has expected approximation factor at most $3\ln 4-9/4+\eps<1.91$. 
Following \cite{BGRS10}, in order to simplify the analysis let us consider the following variant of the algorithm. We introduce a dummy variable $x_{r,r}$ with $p_{r}=0$ corresponding to a dummy component containing the root only, and fix $x_{r,r}$ so that the sum of the $x$'s is some fixed value $M$ in all the iterations. For $M=n^{O(1)}$ large enough, this has no impact on the power of the solution nor on the behaviour of the algorithm (since sampling the dummy component has no effect). Furthermore, we let the algorithm run forever (at some point it will always sample components of power zero).

Let $S^t$ be the min-power Steiner tree at the beginning of iteration $t$ (in particular, $S^1=S^*$). For a given sampled component $C^t$, we let $p(C^t)$ be its power in the considered iteration. We define similarly $p(S^t)$ and the corresponding cost $c(S^t)$. The expected approximation factor of the algorithm is bounded by:
{\small
\begin{align}
\frac{1}{p(S^*)}\sum_{t}E[p(C^t)]  & =\frac{1}{p(S^*)}\sum_{t}\sum_{(Q,s)}E[\frac{x^t_{Q,s}}{M}p_{Q}]\leq \frac{1+\eps}{Mp(S^*)}\sum_{t}E[p(S^t)].\label{eqn:basicApx}
\end{align}
}
\hspace{-4pt}Hence, it is sufficient to provide a good upper bound on $E[p(S^t)]$. We exploit the following high-level (ideal) procedure. We start from $\tilde{S}=S^*$, and at each iteration $t$ we add the sampled component $C^t$ to $\tilde{S}$ and delete some \emph{bridge} edges $B^t$ in $E(\tilde{S})\cap E(S^*)$ in order to remove cycles (while maintaining terminal connectivity). By construction, $\tilde{S}$ is a feasible Steiner tree at any time. Furthermore, the power of $\tilde{S}$ at the beginning of iteration $t$ is equal to the power $p(U^t)$ of the forest of non-deleted edges $U^t$ of $S^*$ at the beginning of the same iteration\footnote{Since edge weights of sampled components are set to zero, any bridge edge can be replaced by a path of zero cost edges which provides the same connectivity.}. In particular, $p(S^t)\leq p(U^t)=\sum_v p_{U^t}(v)$.   

At this point our analysis deviates (and gets slightly more involved) w.r.t. \cite{BGRS10}: in that paper the authors study the expected number of iterations before a given (single) edge is deleted. We rather need to study the behavior of collections of edges incident to a given node $v$. 
%bound $\sum_t E[p_{U^t}(v)]$: to this aim we need to study the behavior of collections of edges incident to $v$ (rather than just one edge). 
In more detail, let $e^1_v,\ldots,e^{d(v)}_v$ be the edges of $S^*$ incident to $v$, in decreasing order of cost  $c^1_v\geq c^2_v\geq,\ldots,\geq c^{d(v)}_v$ (breaking ties arbitrarily). Observe that $p_{U^t}(v)=c^i_v$ during the iterations when all edges $e^1_v,\ldots,e^{i-1}_v$ are deleted and $e^i_v$ is still non-deleted. Define $\delta^i_v$ as the expected number of iterations before all edges $e^1_v,\ldots,e^{i}_v$ are deleted. For notational convenience, define also $\delta^0_v=c^{d(v)+1}_v=0$. Then
{\small
\begin{equation}\label{eqn:newBound}
E[\sum_t p_{U^t}(v)]= \sum_{i=1}^{d(v)}c^i_v(\delta^i_v-\delta^{i-1}_v)
= \sum_{i=1}^{d(v)}\delta^i_v(c^i_v-c^{i+1}_v).
%\leq \sum_t(c^1_v \delta^1_v + \Delta^{\geq 2}\sum_{i=2}^{d(v)}c^i_v).
\end{equation}}
\hspace{-4pt}We will provide a feasible upper bound $\delta^i$ on $\delta^i_v$ for all $v$ (for a proper choice of the bridge edges $B^t$) with the following two properties for all $i$:
{\small
$$
\mathbf{(a)}\; \delta^i\leq \delta^{i+1} \hspace{2cm } \mathbf{(b)}\; \delta^{i}-\delta^{i-1}\geq \delta^{i+1}-\delta^i.
$$}
\hspace{-3pt}In words, the $\delta^i$'s are increasing (which is intuitive since one considers larger sets of edges) but at decreasing speed.
%Observe that $\delta^i_v$ is an increasing function of $i$: this is because the edges corresponding to $\delta^{i+1}_v$ are a superset of the edges corresponding to $\delta^i_v$. We will provide a feasible upper bound $\delta^i$ on $\delta^i_v$ for all $v$ (for a proper choice of the bridge edges) and show that the difference $\delta^i-\delta^{i-1}$ is a decreasing function of $i$. In other terms, the edges $e^i_v$ with smaller cost determine the power of $v$ for a smaller number of iterations. 
Consequently, from \eqref{eqn:newBound} one obtains
{\small
\begin{equation}\label{eqn:newBound2}
E[\sum_t p_{U^t}(v)]\leq \delta^1 c^1_v + \max_{i\geq 2}\{\delta^i-\delta^{i-1}\} \sum_{i=2}^{d(v)}c^i_v = \delta^1 c^1_v + (\delta^2-\delta^1) \sum_{i=2}^{d(v)}c^i_v. 
%\leq \sum_t(c^1_v \delta^1_v + \Delta^{\geq 2}\sum_{i=2}^{d(v)}c^i_v).
\end{equation}}
\hspace{-3pt}Inspired by \eqref{eqn:newBound2}, we introduce the following classification of the edges of $S^*$. We say that the power of node $v$ is \emph{defined} by $e^1_v$. We partition the edges of $S^*$ into the \emph{heavy} edges $H$ which define the power of both their endpoints, the \emph{middle} edges $M$ which define the power of exactly one endpoint, and the remaining \emph{light} edges $L$ which do not define the power of any node. Let $c(H)=\gamma_H\,c(S^*)$ and $c(M)=\gamma_M\,c(S^*)$. Observe that $p(S^*)=\alpha\,c(S^*)$ where $\alpha=2\gamma_H+\gamma_M\in [1,2]$. 
Note also that in \eqref{eqn:newBound2} heavy edges appear twice with coefficient $\delta^1$, middle edges appear once with coefficient $\delta^1$ and once with coefficient $\delta^2-\delta^1$, and light edges appear twice with coefficient $\delta^2-\delta^1$. Therefore one obtains
{\small
\begin{align}
E[\sum_t p(U^t)] & =\sum_v E[\sum_t p^t(v)]  \leq 2\delta^1 c(H) + (\delta^1+\delta^2-\delta^1)c(M)+2(\delta^2-\delta^1)c(L) \notag \\  & = (2\delta^1 \gamma_H +\delta^2 \gamma_M +2(\delta^2-\delta^1)(1-\gamma_H-\gamma_M))\cdot c(S^*) \notag \\ & = \left(2(\delta^2-\delta^1)+(2\delta^1-\delta^2)\alpha\right)\cdot \frac{p(S^*)}{\alpha}\overset{\alpha\geq 1}{\leq} \delta^2\,p(S^*). \label{eqn:newBound3}
\end{align}}
\hspace{-3pt}Summarizing the above discussion, the approximation factor of the algorithm can be bounded by
{\small
\begin{align}
\frac{1+\eps}{Mp(S^*)}\sum_{t}E[p(S^t)]\leq \frac{1+\eps}{Mp(S^*)}\sum_{t}E[p(U^t)]\overset{\eqref{eqn:newBound3}}{\leq}\frac{(1+\eps)\delta^2}{M}.\label{eqn:refinedApx}
\end{align}
}

We next provide the mentioned bounds $\delta^i$ satisfying Properties (a) and (b): we start with the spanning tree case and then move to the more complex and technical general case.

%We start by discussing the spanning tree case. This way we also introduce part of the ideas for the more complex general case, which is discussed next.

%\paragraph{\bf The spanning tree case.} 
\subsection{The Spanning Tree Case.} 

Observe that in this case the optimal solution $T^*:=S^*$ is by definition a terminal spanning tree (i.e. a Steiner tree without Steiner nodes). Therefore we can directly exploit the following claim in \cite{BGRS10}.

\begin{lemma}\cite{BGRS10}\label{lem:deletionSpanning}
Let $T^*$ be any terminal Steiner tree. Set $\tilde{T}:=T^*$ and consider the following process. For $t=1,2,\ldots$: {\bf (a)} Take any feasible solution $x^t$ to $LP_{pow}$; {\bf (b)} Sample one component $C^t=C_{Q}$ with probability proportional to variables $x^t_{Q,s}$; {\bf (c)} Delete a subset of bridge edges $B^t$ from $E(\tilde{T})\cap E(T^*)$ so that all the terminals remain connected in $\tilde{T}-B^t\cup C^t$.  There exists a randomized procedure to choose the $B^t$'s so that
any $W\subseteq E(T^*)$ is deleted after $M\, H_{|W|}$ iterations in expectation\footnote{$H_q:=\sum_{i=1}^{q}\frac{1}{i}$ is the $q$-th harmonic number.}.
\end{lemma}
By Lemma \ref{lem:deletionSpanning} with $W=\{e^1_v,\ldots,e^i_v\}$, we can choose $\delta^i= M\cdot H_i$. 
Observe that these $\delta^i$'s satisfy Properties (a) and (b) since $\frac{\delta^{i+1}-\delta^{i}}{M}=\frac{1}{i+1}$ is a positive decreasing function of $i$.
%$\delta^i=M\cdot H_i\leq M\cdot H_{i+1}=\delta^{i+1}$ and Property (b) since $\delta^i-\delta^{i-1}=\frac{M}{i}\geq \frac{M}{i+1}=\delta^{i+1}-\delta^{i}$. 
%Observe that these $\delta^i$'s satisfy Property (a) since $\delta^i=M\cdot H_i\leq M\cdot H_{i+1}=\delta^{i+1}$ and Property (b) since $\delta^i-\delta^{i-1}=\frac{M}{i}\geq \frac{M}{i+1}=\delta^{i+1}-\delta^{i}$. 
Theorem \ref{thr:mainSpanning} immediately follows by \eqref{eqn:refinedApx} since $\frac{(1+\eps)\delta^2}{M}=\frac{(1+\eps)M\,H_2}{M}=(1+\eps)\cdot \frac{3}{2}$.

%\begin{theorem}\label{thr:mainSpanning}
%Algorithm \ref{alg:mpst} is a $1.5+\eps$ expected approximation for min-power spanning tree.\rem{F: move to intro}
%\end{theorem}
%\begin{proof}
%From Lemma \ref{lem:deletionSpanning}, we can choose $\delta^i= M\cdot H_i$. Observe that the $\delta^i$'s satisfy Property (a) since $\delta^i=M\cdot H_i\leq M\cdot H_{i+1}=\delta^{i+1}$ and Property (b) since $\delta^i-\delta^{i-1}=\frac{M}{i}\geq \frac{M}{i+1}=\delta^{i+1}-\delta^{i}$. Hence \eqref{eqn:refinedApx} holds and the approximation factor is at most $\frac{(1+\eps)\delta^2}{M}=\frac{(1+\eps)M\,H_2}{M}=(1+\eps)\cdot \frac{3}{2}$.
%%
%%Therefore \eqref{eqn:newBound3} holds. Combining \eqref{eqn:basicApx} and \eqref{eqn:newBound3}, the approximation factor is bounded by
%%{\small
%%$$
%%E[\sum_{t}p_{U^t}(v)]= \sum_t \sum_{i=1}^{d(v)}c^i_v(del^i_v-del^{i-1}_v)
%%= \sum_{t}\sum_{i=1}^{d(v)}del^i_v(c^i_v-c^{i+1}_v).
%%$$}
%\end{proof}

%\paragraph{\bf The general case.} 
\subsection{The General Case.}

Here we cannot directly apply Lemma \ref{lem:deletionSpanning} since $S^*$ might not be a terminal spanning tree: w.l.o.g. assume that $S^*$ consists of one full component. Following \cite{BGRS10}, we define a proper auxiliary terminal spanning tree $T^*$, the \emph{witness tree} (see also Figure \ref{fig:witnesstree}). We turn $S^*$ into a rooted binary tree  $S^*_{bin}$ as follows: Split one edge, and root the tree at the newly created node $r$. Split internal nodes of degree larger than $3$ by introducing dummy nodes and dummy edges of cost zero. We make the extra assumption\footnote{This is irrelevant for \cite{BGRS10}, but it is useful in the proof of Lemma \ref{lem:technical}.}, that we perform the latter step so that the $i$ most expensive edges incident to a given node appear in the highest possible (consecutive) levels of $S^*_{bin}$. Finally, shortcut internal nodes of degree $2$.
%More precisely, consider an internal node $v$ with children $u_1,u_2,u_3,\ldots,u_{d}$, sorted in increasing order of costs $c(vu_i)$\footnote{The assumption on the costs is not crucial here, but it will be used later for different purposes}. We remove the edges $vu_i$, add a path $vv_1,\ldots,v_{d-2}$, where the $v_i$'s are dummy nodes and the edges on the path are dummy edges of cost zero. Finally, we add the edges $vu_1$, $v_1u_2$,\ldots, $v_{d-2}u_{d-1}$ and $v_{d-2}u_{d}$, 	where edge $xu_i$ has cost $c(vu_i)$. 
Tree $T^*$ is constructed as follows. For each internal node $v$ in $S^*_{bin}$ with children $u$ and $z$, mark uniformly at random exactly one of the two edges $vu$ and $vz$. Given two terminals $r'$ and $r''$, add $r'r''$ to $T^*$ iff the path between $r'$ and $r''$ in $S^*_{bin}$ contains exactly one marked edge. We associate to each edge $f'\in E(S^*_{bin})$ a (non-empty) \emph{witness set} $W(f')$ of edges of  $T^*$ as follows: $e=uv\in E(T^*)$ belongs to $W(f')$ iff the path between $u$ and $v$ in $S^*_{bin}$ contains $f'$. There is a many-to-one correspondence from each $f\in E(S^*)$ to some $f'\in E(S^*_{bin})$: we let $W(f):=W(f')$. 

We next apply the same deletion procedure as in Lemma \ref{lem:deletionSpanning} to $T^*$. When all the edges in $W(f)$ are deleted, we remove $f$ from $S^*$: this process defines the bridge edges $B^t$ that we remove from $\tilde{S}$ at any iteration.
%When we delete $f$ from $S^*_{bin}$, we also delete the corresponding edges in $S^*$ (zero, one or more edges, depending on the cases). 
As shown in \cite{BGRS10}, the non-deleted edges $U^t$ of $S^*$ at the beginning of iteration $t$ plus  the components which are sampled in the previous iterations induce (deterministically) a feasible Steiner tree. Hence also in this case we can exploit the upper bound $p(S^t)\leq p(U^t)=\sum_v p_{U^t}(v)$. 
Let us define $W^i(v):=\cup_{j=1}^{i}W(e^j_v)$. In particular, in order to delete all the edges $e^1_v,\ldots,e^i_v$ we need to delete $W^i(v)$ from $T^*$. 
%It then follows from Lemma \ref{lem:deletionSpanning} that:
%{\small 
%$$
%\delta^i_v \leq \sum_{q\geq 1} Pr[ |W^i(v)| = q] \cdot M\,H_q. 
%$$}
The next technical lemma provides a bound on $\delta^i_v$ by combining Lemma \ref{lem:deletionSpanning} with an analysis of the distribution of $|W^i(v)|$. The crucial intuition here is that sets $W(e^j_v)$ are strongly correlated and hence $|W^i(v)|$ tends to be much smaller than $\sum_{j=1}^{i}|W(e^j_v)|$.  
\begin{lemma}\label{lem:technical}
$\delta^i_v\leq \delta^i:=\frac{1}{2^i}MH_i+(1-\frac{1}{2^i})\sum_{q\geq 1}\frac{1}{2^q}MH_{q+i}$.
\end{lemma}
\begin{proof}  %{\em (sketch)}
%{\bf (Lemma \ref{lem:technical})}
Let us assume that $v$ has degree $d(v)\geq 3$ and that $i<d(v)$, the other cases being analogous and simpler. Recall that we need to delete all the edges in $W^i(v)$ in order to delete $e^1_v,\ldots,e^i_v$, and this takes time $M\,H_{|W^i(v)|}$ in expectation by Lemma \ref{lem:deletionSpanning}. Let us study the distribution of $|W^i(v)|$. Consider the subtree $T'$ of $S^*_{bin}$ given by (the edges corresponding to) $e^1_v,\ldots,e^i_v$ plus their sibling (possibly dummy) edges. Observe that, by our assumption on the structure of $S^*_{bin}$, this tree has $i+1$ leaves and height $i$. We expand $T'$ by appending to each leaf $v$ of $T'$ the only path of unmarked edges from $v$ down to some leaf $\ell(v)$ of $S^*_{bin}$: let $C'$ be the resulting tree (with $i+1$ leaves). Each edge of $T^*$ with both endpoints in (the leaves of) $C'$ is a witness edge in $W^i(v)$. The number of these edges is at most $i$ since the witness tree is acyclic: assume pessimistically that they are exactly $i$. Let $r'$ be the root of $T'$, and $s'$ be the (only) leaf of $T'$ such that the edges on the path from $r'$ to $s'$ are unmarked. Let also $d'$ be the only leaf of $T'$ which is not the endpoint of any $e^j_v$, $j\leq i$ ($d'$ is defined since $i<d(v)$). Observe that $Pr[s'=d']=1/2^i$ since this event happens only if the $i$ edges along the path from $r'$ to $d'$ are unmarked. When $s'\neq d'$, there are st most $|E(P')|+1$ extra edges in $W^i(v)$, where $P'$ is a maximal path of unmarked edges starting from $r'$ and going to the root of $S^*_{bin}$. If $h_{i,v}$ is the maximum value of $|E(P')|$, then $Pr[|E(P')|=q]=1/2^{\min\{q+1,h_{i,v}\}}$ for $q\in [0,h_{i,v}]$. Altogether:
%By Lemma \ref{lem:deletionSpanning}, deleting the edges $e^1_v,\ldots,e^i_v$ takes $M\cdot H_{|W^i(v)|}$ iterations in expectation, hence $\delta^i_v \leq \sum_{g\geq 1} Pr[|W^i(v)|=g]\cdot MH_{g}$. 
{\small
\begin{align*}
\delta^i_v & \overset{Lem. \ref{lem:deletionSpanning}}\leq \sum_{g\geq 1} Pr[|W^i(v)|=g]\cdot MH_{g} \leq \frac{MH_i}{2^i} + (1-\frac{1}{2^i})\cdot \sum_{q=0}^{h_{i,v}}\frac{MH_{i+q+1}}{2^{\min\{q+1,h_{i,v}\}}} \leq \delta^i.\hspace{18pt}\qed 
\end{align*}}
\end{proof}
The reader may check that the above $\delta^i$'s satisfy Properties (a) and (b) since 
{\small
$$
\frac{\delta^{i+1}-\delta^{i}}{M}=\frac{1}{(i+1)2^{i+1}}+(1-\frac{1}{2^{i+1}})\sum_{q\geq 1}\frac{1}{2^q(q+i+1)}+\frac{1}{2^{i+1}}(\sum_{q\geq 1}\frac{H_{q+i}}{2^q}-H_i)
$$} 
\hspace{-3pt}is a positive decreasing function of $i$. Theorem \ref{thr:mainSteiner} immediately follows from \eqref{eqn:refinedApx} and Lemma \ref{lem:technical} since  $\delta^2=\frac{H_2}{4}+3(\sum_{q\geq 1}\frac{H_q}{2^q}-\frac{H_1}{2}-\frac{H_2}{4})=3\ln 4 - \frac{9}{4}$. 

%%%%%%%
%PROOF of (b) 
%%%%%%%
%{\small
%\begin{align*}
%& 2^i\cdot (2b_{i+1}-b_{i+2}-b_{i}) \\
%= & H_{i+1}-\frac{H_{i+2}}{4}-H_i+\sum_{q\geq 1}\left( (2^{i+1}-1)\frac{H_{i+1+q}}{2^q}- (2^{i+2}-1)\frac{H_{i+2+q}}{2^{q+2}}-(2^i-1)\frac{H_{i+q}}{2^q}\right)\\
%\geq & \frac{1}{i+1}+\sum_{q\geq 1}\left( (2^{i+1}-1)\frac{H_{i+1+q}}{2^q}- 2^{i}\frac{H_{i+2+q}}{2^q}-(2^i-1)\frac{H_{i+q}}{2^q}\right)\\
%\geq & \frac{1}{i+1}-\sum_{q\geq 1}\frac{1}{2^q}\frac{1}{i+1+q}\geq 0. %\quad\quad 
%%\Rightarrow \quad\quad b_{i+2}-b_{i+1} \leq b_{i+1}-b_{i}.
%\end{align*}
%}

\paragraph{\bf Acknowledgments.} We thank Marek Cygan for reading a preliminary version of this paper and Zeev Nutov for brinding the min-power spanning tree problem to our attention (during a Dagstuhl workshop) and for mentioning some analogies between that problem and min-cost Steiner tree.

{\small
\bibliographystyle{plain}
\bibliography{bibminpower.bib}
}

\newpage

\section*{Appendix}

\begin{Figure}[t] %[H]
\begin{center}
\psset{xunit=0.52cm, yunit=0.9cm, labelsep=2pt}
\begin{pspicture}(9,-0.5)(9,3.0)
\psset{labelsep=0, linewidth=1.0pt, framesize=6pt}
\psset{labelsep=0pt}
\fnode(-1.0,0){a}\rput[lc](-1.0,-0.3){a} % terminals
\fnode(0.0,0){b}\rput[lc](0.0,-0.3){b} 
\fnode(2.0,0){c}\rput[lc](2.1,-0.3){c} 
%\fnode(3.0,0){e}\rput[lc](3.1,-0.3){e} 
\cnode(0.0,1){3pt}{s4}
\cnode(1.0,1){3pt}{s2}
\fnode(2.2,1){d}\rput[lc](2.2,0.7){d} 
%\cnode(3.0,1){2pt}{s3}
\cnode(1.0,2){3pt}{s1}\rput[lc](0.6,2.0){v}
\fnode(0.0,2){e}\rput[lc](-0.5,2){e} 
\fnode(1.0,3.0){f}\rput[lc](1.5,3.0){f} 
%\fnode(2.0,2){f}\rput[lc](2.5,2){f} 
\cnode(1.0,2.5){3pt}{r}%\rput[lc](1.0,2.7){r}
\ncline{s4}{a}\nbput{$6$}
\ncline[linewidth=2pt]{s4}{b}\rput(0.4,0.3){$3$}%\naput{$6$}
\ncline{s2}{c}\nbput{$8$}
\ncline[linewidth=2pt]{s1}{s4}\nbput{$1$}
\ncline{s1}{s2}\naput{$2$}%\rput(0.7,1.3){$2$}%\nbput{$2$}
\ncline{s1}{d}\naput{$4$}%\rput(2.0,1.3){$4$}%\nbput{$3$}
%\ncline{s1}{s3}\rput(2.7,1.4){$5$}%\naput{$5$}
\ncline{r}{s1}
\ncline{r}{e}
\ncline{r}{f}
\psellipse[linecolor=gray](0.5,1)(1,0.3){}
\psellipse[linecolor=gray](2.2,1)(0.5,0.3){}
%\ncline{s3}{e}\naput{$3$}
%\rput(5,1){$\Rightarrow$}
\end{pspicture}
%%%%%%%%%%
%%%%%%%%%%
%%%%%%%%%%
\begin{pspicture}(2,-0.5)(2,3.0)
\psset{labelsep=0, linewidth=1.0pt, framesize=6pt}
\psset{labelsep=0pt}
\fnode(0.0,0){a}\rput[lc](0.0,-0.3){a} % terminals
\fnode(1.0,0){b}\rput[lc](1.0,-0.3){b} 
\fnode(2.0,0){c}\rput[lc](2.0,-0.3){c} 
\cnode(0.0,1){3pt}{s4}
\cnode(2.0,1){3pt}{s2}
\cnode(1.0,2){3pt}{s1}\rput[lc](0.6,2.0){v}
\ncline{s4}{a}%\nbput{$6$}
\ncline{s4}{b}%\rput(0.4,0.3){$3$}%\naput{$6$}
\ncline{s2}{c}%\nbput{$8$}
\ncline{s1}{s4}%\nbput{$1$}
\ncline{s1}{s2}
\psellipse[linestyle=dashed](2,0.4)(0.5,1.0){}
\end{pspicture}
\begin{pspicture}(-2.0,-0.5)(-2.0,3.5)
\psset{labelsep=0, linewidth=1.0pt, framesize=6pt}
\psset{labelsep=0pt}
\fnode(0.0,0){b}\rput[lc](0.0,-0.3){b} 
\cnode(0.0,1){3pt}{s4}
\fnode(2.0,1){dbis}\rput[lc](2.0,0.7){d} 
\cnode(1.0,2){3pt}{s1bis}\rput[lc](0.6,2.0){v}
\fnode(0.0,2){e}\rput[lc](-0.5,2){e} 
\fnode(1.0,3.0){f}\rput[lc](1.5,3.0){f} 
\cnode(1.0,2.5){3pt}{r}%\rput[lc](1.0,2.7){r}
\ncline{s4}{b}%\nbput{$3$}
\ncline{s1bis}{s4}%\nbput{$1$}
\ncline{s1bis}{dbis}%\naput{$4$}%\rput(1.4,1.3){$4$}%\nbput{$3$}
\ncline{r}{s1bis}
\ncline{r}{e}
\ncline{r}{f}
\psellipse[linestyle=dashed](0,0.4)(0.5,1.0){}
\psline[linestyle=dashed]{->}(-0.5,0.5)(-1.8,0.5)\rput(-1.1,0.7){$6$}
%\nccurve[linecolor=gray, angleA=-45, angleB=-135]{s1}{dbis}
\ncline[linecolor=gray,linestyle=dashed]{->}{s1}{dbis}\naput{2}
\end{pspicture}
\begin{pspicture}(-8.0,-0.2)(-8.0,3.3)
\psset{labelsep=0, linewidth=1.0pt, framesize=6pt}
\psset{labelsep=0pt}
\fnode(0.0,0){a}\rput[lc](0.0,-0.3){a} 
\fnode(2.0,0){c}\rput[lc](2.0,-0.3){c} 
\cnode(1,0.75){3pt}{sabc}
\fnode(1.0,1.5){b}\rput[lc](1.5,1.6){b} 
\cnode(1,2.25){3pt}{sbdef}
\fnode(1.0,3){f}\rput[lc](1.0,3.3){f} 
\fnode(0.0,2.25){e}\rput[lc](0.0,2.6){e} 
\fnode(2.0,2.25){d}\rput[lc](2.0,2.6){d} 
\ncline{sabc}{a}
\ncline{sabc}{b}
\ncline{sabc}{c}
\ncline{sbdef}{b}
\ncline{sbdef}{d}
\ncline{sbdef}{e}
\ncline{sbdef}{f}
\end{pspicture}
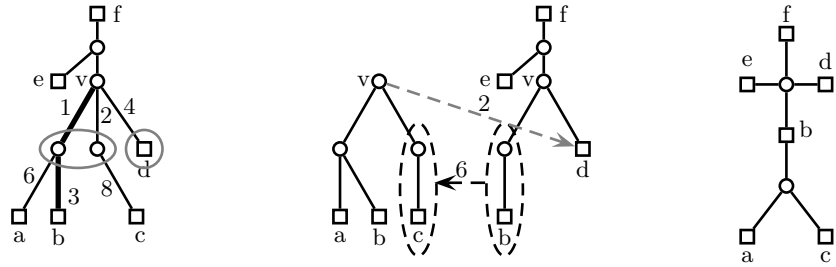
\caption{Decomposition of a tree into components of maximum degree $\Delta=3$ as in Lemma \ref{lem:boundedDegree}. {\bf (Left)} Part of the edges are labelled with their weight. Squares denote terminals. The chosen root is $f$. The only split node is $v$. The corresponding sets of children $V_1$ and $V_2$, for $\Delta'=2$, are indicated by the gray ovals. Bold edges denote the path $P_{m(1)}=P_1$ associated with $V_1$. {\bf (Middle)} The resulting two components. The gray dashed arrow illustrates the charging for the new copy of node $v$, and the black dashed arrow the charging of the nodes in $P_{1}-\{v\}$ (dashed oval on the right) to the nodes in $P_2-\{v\}$ (dashed oval on the left). {\bf (Right)} The corresponding component graph.}
\label{fig:lemDecomposition} 
\end{center}
\end{Figure}

\begin{Figure}[t] %[H]
\begin{center}
\psset{xunit=0.52cm, yunit=0.9cm, labelsep=2pt}
\begin{pspicture}(10.5,-0.2)(10.5,3.0)
\psset{labelsep=0, linewidth=1.0pt, framesize=6pt}
\psset{labelsep=0pt}
\cnode(2.5,2.5){3pt}{r}\rput[lc](2.5,2.8){r}
\cnode[fillstyle=solid,fillcolor=black](1.5,2){3pt}{s1}
\cnode[fillstyle=solid,fillcolor=black](3.5,2){3pt}{s2} \rput[lc](3.6,2.3){v}
\cnode(1.0,1.5){3pt}{s3}
\fnode(2.0,1.5){a}\rput[lc](2.0,1.25){a}
\fnode(3.0,1.5){b}\rput[lc](3.0,1.25){b}
\cnode(4.0,1.5){3pt}{s6}\rput[lc](4.3,1.7){u}
\cnode(0.5,1.0){3pt}{s7}
\fnode(1.5,1.0){c}\rput[lc](1.5,0.75){c}
\cnode(3.5,1.0){3pt}{s9}
\fnode(4.5,1.0){d}\rput[lc](4.5,0.75){d}
\fnode[fillstyle=solid,fillcolor=black](0.0,0.5){e}\rput[lc](0,0.25){e}
\fnode[fillstyle=solid,fillcolor=black](1.0,0.5){f}\rput[lc](1,0.25){f}
\cnode[fillstyle=solid,fillcolor=black](3.0,0.5){3pt}{s13}
\fnode[fillstyle=solid,fillcolor=black](4.0,0.5){g}\rput[lc](4.0,0.25){g}
\fnode(2.5,0.0){h}\rput[lc](2.5,-0.25){h}
\fnode(3.5,0.0){i}\rput[lc](3.5,-0.25){i}
\ncline{r}{s1}
\ncline{r}{s2}
\ncline{s1}{s3}
\ncline{s1}{a}
\ncline{s2}{b}
\ncline[linewidth=2.5pt]{s2}{s6}
\ncline{s3}{s7}
\ncline{s3}{c}
\ncline[linewidth=2.5pt]{s6}{s9}
\ncline{s6}{d}
\ncline{s7}{e}
\ncline{s7}{f}
\ncline[linewidth=2.5pt]{s9}{s13}
\ncline{s9}{g}
\ncline[linewidth=2.5pt]{s13}{h}
\ncline{s13}{i}
\psline[linestyle=dashed](-0.5,2)(5.0,2)
\psline[linestyle=dashed](-0.5,0.5)(5.0,0.5)
\end{pspicture}
%%%
%%%%
%%%
\begin{pspicture}(3.7,-0.2)(3.7,3.0)
\psset{labelsep=0, linewidth=1.0pt, framesize=6pt}
\psset{labelsep=0pt}
\cnode(1.75,2.5){3pt}{r}
\cnode(1.0,2){3pt}{s1}
\cnode(2.5,2){3pt}{s2}
\fnode(1.5,1.5){a}\rput[lc](1.5,1.25){a}
\cnode(3.0,1.5){3pt}{s6}\rput[lc](3.3,1.7){u}
\cnode(2.5,1.0){3pt}{s9}
\cnode(2.0,0.5){3pt}{s13}
\fnode(1.5,0.0){h}\rput[lc](1.5,-0.25){h}
\ncline{r}{s1}
\ncline{r}{s2}
\ncline{s1}{a}
\ncline{s2}{s6}
\ncline{s6}{s9}
\ncline{s9}{s13}
\ncline{s13}{h}
\end{pspicture}
%%%
%%%
%%%
\begin{pspicture}(0,-0.2)(0,2.5)
\psset{labelsep=0, linewidth=1.0pt, framesize=6pt}
\psset{labelsep=0pt}
\cnode(1.5,2.5){3pt}{s1}
\cnode(1.0,2.0){3pt}{s3}
\fnode(2.0,2.0){a}\rput[lc](2.0,1.75){a}
\cnode(0.5,1.5){3pt}{s7}
\fnode(1.5,1.5){c}\rput[lc](1.5,1.25){c}
\fnode(0.0,1.0){e}\rput[lc](0,0.75){e}
\fnode(1.0,1.0){f}\rput[lc](1,0.75){f}
\ncline{s1}{s3}
\ncline{s1}{a}
\ncline{s3}{s7}
\ncline{s3}{c}
\ncline{s7}{e}
\ncline{s7}{f}
\cnode(2.0,0.5){3pt}{s13}
\fnode(1.5,0.0){h}\rput[lc](1.5,-0.25){h}
\fnode(2.5,0.0){i}\rput[lc](2.5,-0.25){i}
\ncline{s13}{h}
\ncline{s13}{i}
\end{pspicture}
\begin{pspicture}(-0.1,-0.2)(-0.1,3.0)
\psset{labelsep=0, linewidth=1.0pt, framesize=6pt}
\psset{labelsep=0pt}
\cnode(3.5,2){3pt}{s2} 
\fnode(3.0,1.5){b}\rput[lc](3.0,1.25){b}
\cnode(4.0,1.5){3pt}{s6}\rput[lc](4.3,1.7){u}
\cnode(3.5,1.0){3pt}{s9}
\fnode(4.5,1.0){d}\rput[lc](4.5,0.75){d}
\cnode(3.0,0.5){3pt}{s13}
\fnode(4.0,0.5){g}\rput[lc](4.0,0.25){g}
\fnode(3.5,0.0){i}\rput[lc](3.5,-0.25){i}
\ncline{s2}{b}
\ncline{s2}{s6}
\ncline{s6}{s9}
\ncline{s6}{d}
\ncline{s9}{s13}
\ncline{s9}{g}
\ncline{s13}{i}
\end{pspicture}
%%%%%%%%
%%%%%%%%
%%%%%%%%
\begin{pspicture}(-10.5,-0.3)(-10.5,3.0)
\psset{labelsep=0, linewidth=1.0pt, framesize=6pt}
\psset{labelsep=0pt}
\fnode(0,1){h}\rput[lc](0.5,1){h}
\cnode(0,1.5){3pt}{sha}
\fnode(0,2){a}\rput[lc](0.5,2){i}
\cnode(0,0.5){3pt}{shi}
\fnode(0,0){i}\rput[lc](0.5,0){a}
\cnode(-1,0){3pt}{scefi}
\fnode(-2,0){e}\rput[lc](-2.5,0){e}
\fnode(-1,0.5){f}\rput[lc](-1.5,0.5){f}
\fnode(-1,-0.5){c}\rput[lc](-1.5,-0.5){c}
\cnode(-1,2){3pt}{sbdga}
\fnode(-2,2){d}\rput[lc](-2.5,2){d}
\fnode(-1,2.5){b}\rput[lc](-1.5,2.5){b}
\fnode(-1,1.5){g}\rput[lc](-1.5,1.5){g}
\ncline{sha}{a}
\ncline{sha}{h}
\ncline{shi}{h}
\ncline{shi}{i}
\ncline{scefi}{i}
\ncline{scefi}{c}
\ncline{scefi}{e}
\ncline{scefi}{f}
\ncline{sbdga}{b}
\ncline{sbdga}{d}
\ncline{sbdga}{g}
\ncline{sbdga}{a}
\end{pspicture}
\caption{Decomposition of a component as in Theorem \ref{thr:decomposition}. {\bf (Left)} A component $C$  after contracting nodes of degree $2$ other than the root $r$. Squares denote terminals and black nodes are marked in the case $q=1$. Dashed lines suggest the partition of $C$ into edge-disjoint subtrees. Bold edges indicate the path $P(v)$. {\bf (Middle)} The 
resulting set of components $C_T$: regular edges indicate the  subtree $T$ associated to $C_T$, and bold edges the paths $P(w)$ associated to the leaves $w$ of $T$. There are two components containing $u$: the left one because the left child of $u$ is along the path $P(v)$ of marked node $v$, and the right one because $u$ belongs to the subtree $T$ of $v$. {\bf (Right)} The corresponding component graph.}
\label{fig:thrDecomposition} 
\end{center}
\end{Figure}

\begin{Figure}[t] %[H]
\begin{center}
\psset{xunit=0.52cm, yunit=0.9cm, labelsep=2pt}
\begin{pspicture}(10.0,-1.5)(10.0,3.0)
\psset{labelsep=0, linewidth=1.0pt, framesize=6pt}
\psset{labelsep=0pt}
\cnode(1.0,2){3pt}{s1}\rput[lc](1.0,2.3){v}
\cnode(3,2){3pt}{s2}
\cnode(3,1){3pt}{s3}
\cnode(0,1){3pt}{sl} % terminals
\fnode(-1,0){a}\rput[lc](-1,-0.3){a} % terminals
\fnode(1,0){b}\rput[lc](1,-0.3){b} % terminals
\fnode(1.0,1){c}\rput[lc](1.0,0.7){c} 
\fnode(2,1){d}\rput[lc](2,0.7){d} 
\fnode(3,0){e}\rput[lc](3,-0.3){e} 
\fnode(4.0,1){f}\rput[lc](4.0,0.7){f} 
\ncline{s1}{sl} \nbput{$1$}
\ncline{s1}{c} \nbput{$5$}
\ncline{s1}{d} \naput{$8$}
\ncline{s1}{s2} \naput{$2$}
\ncline{sl}{a} \nbput{$9$}
\ncline{sl}{b} \nbput{$7$}
\ncline{s2}{s3} \nbput{$3$}
\ncline{s2}{f} \naput{$6$}
\ncline{s3}{e} \nbput{$4$}
%\rput(5,1){$\Rightarrow$}
\end{pspicture}
\begin{pspicture}(3,-1.5)(3,3.0)
\psset{labelsep=0, linewidth=1.0pt, framesize=6pt}
\psset{labelsep=0pt}
\cnode(0,0){3pt}{s5}
\fnode(-0.5,-1){a}\rput[lc](-0.5,-1.3){a} % terminals
\fnode(0.5,-1){b}\rput[lc](0.5,-1.3){b} % terminals
\fnode(1.5,0){c}\rput[lc](1.6,-0.3){c} 
\fnode(2.5,1){d}\rput[lc](2.6,0.7){d} 
\cnode(3.0,2.5){3pt}{r}\rput[lc](3.0,2.8){r}
\cnode(1.5,2){3pt}{s1}%\rput[lc](1.4,2.3){r'}
\cnode(4.5,2){3pt}{s2}
\cnode(0.75,1){3pt}{s4}
\fnode(3.5,1){e}\rput[lc](3.5,0.7){e} 
\fnode(5.5,1){f}\rput[lc](5.6,0.7){f} 
\ncline[linewidth=2pt]{s5}{a} \nbput{$9$}
\ncline{s5}{b} \naput{$7$}
\ncline[linewidth=2pt]{s4}{s5} \nbput{$1$}
\ncline{c}{s4} \nbput{$5$}
\ncline{s1}{s4} \nbput{$0$}
\ncline[linewidth=2pt]{d}{s1} \nbput{$8$}
\ncline[linewidth=2pt]{r}{s1} \nbput{$2$}
\ncline{r}{s2} \naput{$0$}
%\ncline{s2}{s3} \nbput{$3$}
\ncline[linewidth=2pt]{s2}{e} \nbput{$7$}
\ncline{s2}{f} \naput{$6$}
\nccurve[linecolor=gray, angleA=-45, angleB=-135]{a}{b}
\nccurve[linecolor=gray, angleA=-45, angleB=-100]{b}{c}
\nccurve[linecolor=gray, angleA=-60, angleB=-100]{c}{d}
\nccurve[linecolor=gray, angleA=-55, angleB=-135]{c}{f}
\nccurve[linecolor=gray, angleA=-45, angleB=-135]{e}{f}
\end{pspicture}
%%%%%%%
%%%%%%%
%%%%%%%
\begin{pspicture}(-6,-1.5)(-6,3.0)
\psset{labelsep=0, linewidth=1.0pt, framesize=6pt}
\psset{labelsep=0pt}
\cnode(0,0){3pt}{s5}\rput[lc](-0.2,0.3){d'}
%\fnode(-0.5,-1){a}\rput[lc](-0.5,-1.3){a} % terminals
\fnode(0.5,-1){b}%\rput[lc](0.5,-1.3){b} % terminals
\fnode(1.5,0){c}\rput[lc](1.7,0.3){s'} 
\fnode(2.5,1){d}%\rput[lc](2.6,0.7){d} 
%\cnode(3.0,2.5){3pt}{r}
\cnode(1.5,2){3pt}{s1}\rput[lc](1.4,2.3){r'}
%\cnode(4.5,2){3pt}{s2}
\cnode(0.75,1){3pt}{s4}
%\fnode(3.5,1){e}\rput[lc](3.5,0.7){e} 
\fnode(5.5,1){f}%\rput[lc](5.6,0.7){f} 
%\ncline[linewidth=2pt]{s5}{a} \nbput{$9$}
\ncline[linewidth=2pt]{s5}{b} %\naput{$7$}
\ncline{s4}{s5} %\nbput{$1$}
\ncline{c}{s4} %\nbput{$5$}
\ncline{s1}{s4} %\nbput{$0$}
\ncline{d}{s1} %\nbput{$8$}
%\ncline[linestyle=dashed]{r}{s1}
%\ncline[linewidth=2pt]{r}{s2} \naput{$0$}
%\ncline[linewidth=2pt]{s2}{e} \nbput{$7$}
%\ncline{s2}{f} \naput{$6$}
%\nccurve[linecolor=gray, angleA=-45, angleB=-135]{a}{b}
\nccurve[linecolor=gray, angleA=-45, angleB=-100]{b}{c}
\nccurve[linecolor=gray, angleA=-60, angleB=-100]{c}{d}
\nccurve[linecolor=gray, angleA=-55, angleB=-135]{c}{f}
%\nccurve[linecolor=gray, angleA=-45, angleB=-135]{e}{f}
\end{pspicture}
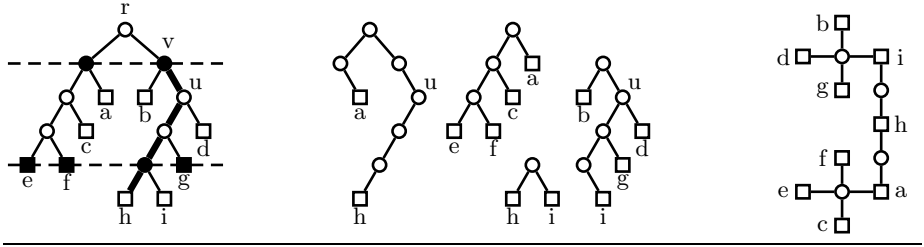
\caption{{\bf (Left)} A Steiner tree $S^*$. Squares denote terminals. Edges are labelled with their costs. {\bf (Middle)} The corresponding binary tree $S^*_{bin}$.  Bold edges are marked. Gray edges define the witness tree $T^*$. The witness sets for the edges of cost $8$ and $5$ are $\{cd\}$ and $\{bc,cd,cf\}$, respectively. Note that these sets have a non-empty intersection. {\bf (Right)} Black edges denote the subtree $C'$ associated with 
the two most expensive edges incident to $v$, of weight $8$ and $5$: regular edges denote $T'$ and bold edges the paths of unmarked edges from the leaves of $T'$ to terminals. The picture also shows the nodes $r'$, $s'$, and $d'$ of $T'$. The corresponding witness set is $W^2(v)=\{bc,cd,cf\}$. Edges $bc,cd\in W^2(v)$ have both endpoints among the leaves of $C'$. In the example $d'\neq s'$ and $P'$ has length $0$ (since the edge from $r'$ to its parent is marked): this corresponds to one extra edge $cf\in W^2(v)$.}  
\label{fig:witnesstree} 
\end{center}
\end{Figure}

\begin{proof} {\em (Lemma \ref{lem:polyLP})}
Consider the optimal fractional solution $x^*$. We define a feasible fractional solution $x'$ where $x'_{Q,s}=0$ for $|Q|>k$. Initially $x'=0$. For any $x^*_{Q,s}$, apply the Decomposition Theorem \ref{thr:decomposition} to $Q$, 
hence obtaining a collection of $k$-components $C_1,\ldots,C_h$. 
Direct the edges in the component graph towards $s$, so as to identify a sink node $s_i$ for each $C_i$. For each $i$, increase $x'_{R\cap V(C_i),s_i}$ by $x^*_{Q,s}$. For a constant $k$ large enough, $x'$ costs at most $1+\eps$ times more than $x^*$. 
Consequently, in order to compute a $1+\eps$ approximate solution, it is sufficient to consider the pairs $(Q,s)$ with $|Q|\leq k$, which are polynomially many. The number of constraints remains exponential, however the separation problem can be solved in polynomial time by the same reduction to MinCut as in \cite{BGRS10}.\qed    
\end{proof}

\begin{proof} {\em (Lemma \ref{lem:time})}
Each iteration takes polynomial time. At any given iteration $t$, if there is no Steiner tree of zero-cost edges, there exists some terminal $r'\neq r$ such that $\sum_{(Q,s): s\in Q \subseteq R,s\neq r',r'\in Q}x^t_{Q,s}\geq 1$ and $p_{Q}>0$ for all the considered $Q$. Since w.l.o.g. $x^t_{Q,s}\leq 1$ and hence $1\leq \sum_{(Q,s)}x^t_{Q,s} \leq n^{O(1)}$, with probability at least $1/n^{O(1)}$ in the current iteration we set to zero the cost of some edge. The claim follows.\qed
\end{proof}

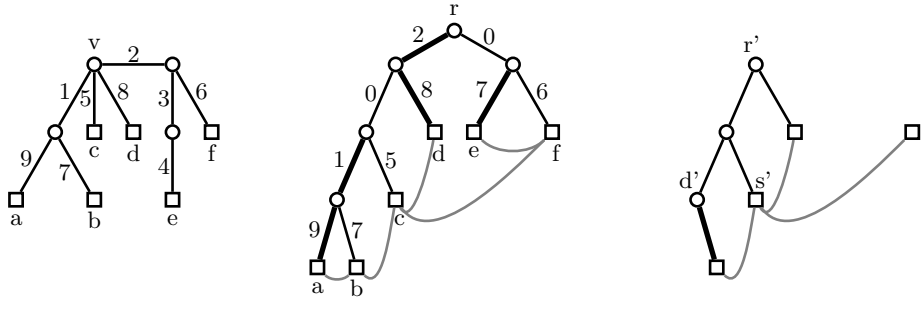
\end{document}